\documentclass{article}
\usepackage{amsmath,amssymb,theorem,enumerate}
\usepackage[dvips]{graphicx,color}
\usepackage[margin=1in]{geometry}

\makeatletter
\newenvironment{proof}[1][\proofname]{\par\normalfont
  \topsep6pt plus6pt\trivlist\item[\hskip\labelsep\itshape
  #1\@addpunct{:}]\ignorespaces}{\qed\endtrivlist}
\newcommand{\proofname}{Proof}
\DeclareRobustCommand{\qed}{%
  \ifmmode
  \else\leavevmode\unskip\penalty9999\hbox{}\nobreak\hfill\fi
  \quad\hbox{\qedsymbol}}
\newcommand{\qedsymbol}{\openbox}
\newcommand{\openbox}{\leavevmode\hbox to.77778em{%
    \hfil\vrule\vbox to.675em{%
      \hrule width.6em\vfil\hrule}\vrule\hfil}}
\makeatother

\theorembodyfont{\itshape}

\newtheorem{proposition}{Proposition}
\newtheorem{theorem}{Theorem}
\theorembodyfont{\rmfamily}
\newtheorem{example}{Example}
\newtheorem{definition}{Definition}
\newtheorem{remark}{Remark}

\newcommand{\Prb}{\mathsf{P}}\newcommand{\Exp}{\mathsf{E}}
\newcommand{\Var}{\mathsf{Var}}
\newcommand{\C}{\mathbb{C}}\newcommand{\N}{\mathbb{N}}
\newcommand{\R}{\mathbb{R}}
\newcommand{\dd}{\mathrm{d}}\newcommand{\ee}{\mathrm{e}}
\newcommand{\B}{\mathcal{B}}
\newcommand{\K}{\mathcal{K}}
\newcommand{\GF}{\mathcal{G}}\newcommand{\Lpl}{\mathcal{L}}
\newcommand{\SINR}{\mathsf{SINR}}\newcommand{\SIR}{\mathsf{SIR}}
\newcommand{\Gam}{\mathrm{Gam}}
\newcommand{\ind}[1]{\boldsymbol{1}_{#1}}
\newcommand{\eqd}{\mathrel{\buildrel d \over =}}
\let\Bar\overline\let\Tilde\widetilde

\title{Spatial modeling and analysis of cellular networks\\ using
  the Ginibre point process: A tutorial}
\author{Naoto Miyoshi$^\dagger$ \and Tomoyuki Shirai$^\ddagger$}
\date{$^\dagger$Department of Mathematical and Computing Science,
  Tokyo Institute of Technology\\
  $^\ddagger$Institute of Mathematics for
  Industry, Kyushu University}

\begin{document}\sloppy\allowbreak\allowdisplaybreaks
\maketitle
\begin{abstract}
Spatial stochastic models have been much used for performance analysis
of wireless communication networks.
This is due to the fact that the performance of wireless networks
depends on the spatial configuration of wireless nodes and the
irregularity of node locations in a real wireless network can be
captured by a spatial point process.
Most works on such spatial stochastic models of wireless networks have
adopted homogeneous Poisson point processes as the models of wireless
node locations.
While this adoption makes the models analytically tractable, it
assumes that the wireless nodes are located independently of each
other and their spatial correlation is ignored.
Recently, the authors have proposed to adopt the Ginibre point
process---one of the determinantal point processes---as the deployment
models of base stations (BSs) in cellular networks.
The determinantal point processes constitute a class of repulsive 
point processes and have been attracting attention due to their
mathematically interesting properties and efficient simulation methods.
In this tutorial, we provide a brief guide to the Ginibre point
process and its variant, $\alpha$-Ginibre point process, as the models
of BS deployments in cellular networks and show some existing results
on the performance analysis of cellular network models with
$\alpha$-Ginibre deployed BSs.
The authors hope the readers to use such point processes as a tool for
analyzing various problems arising in future cellular networks.
\end{abstract}

\section{Introduction}

Spatial stochastic models have been much used for performance analysis
of wireless communication networks and the volume of the literature
has been increasing rapidly, where the wireless nodes are located at
random on the two dimensional Euclidean plane according to some
stochastic point processes (see, e.g., the tutorial articles
\cite{AndrGantHaenJindWebe10,HaenAndrBaccDousFran09,ElsaHossHaen13,ElsaSultAlouWin16}
and monographs \cite{BaccBlas09a,BaccBlas09b,Haen13,Mukh14}).
This is due to the fact that the performance of wireless networks
critically depends on the spatial configuration of wireless nodes and
the irregularity of node locations in a real wireless network can be
well captured by a spatial point process.
Even for cellular networks, many researchers have proposed and
analyzed the spatial stochastic models to cope with various problems
arising from the current explosive growth of mobile data traffic, such
as cognitive radio~\cite{ElsaHoss14}, interference
cancellation~\cite{WildQuekKounRabbSlum14} and so on (a thorough
survey on recent progress is found in \cite{ElsaSultAlouWin16}).

Most works on such spatial stochastic models of wireless networks have
adopted homogeneous Poisson point processes as the models of
wireless node locations and this has been the case for the cellular
networks (see, e.g.,
\cite{AndrBaccGant11,DhilGantBaccAndr12,JoSangXiaAndr12,Mukh12,DiReGuidCora13,ElsaHoss14,WildQuekKounRabbSlum14}).
While this adoption makes the models analytically tractable, it
assumes that the wireless nodes are located independently of each
other and their spatial correlation is ignored.
On the other hand, the base stations~(BSs), in particular macro BSs,
in a cellular network tend to be deployed rather systematically, such
that any two BSs are not too close, and thus a spatial model based on
a point process with repulsive nature seems more
desirable~(see \cite{GuoHaen13}).
Recently, the authors have proposed to adopt the Ginibre point process
and its variant, $\alpha$-Ginibre point process, as the models of BS
deployments in cellular networks and have derived some analytical and
numerical results
(\cite{MiyoShir14a,MiyoShir14b,NakaMiyo14,KobaMiyo14,NagaMiyoShir14,MiyoShir15,KobaMiyo16}).
The Ginibre point process is known as a main example of the
determinantal point processes, which constitute a class of repulsive
point processes and have been attracting attention due to their
mathematically interesting properties and efficient simulation methods
(see, e.g., \cite{Sosh00,ShirTaka03,HougKrisPereVira09,LavaMollRuba15}
for details).
The $\alpha$-Ginibre point process is also one of the determinantal
point processes and is introduced in \cite{Gold10} for interpolating
between the original Ginibre and homogeneous Poisson point processes
by a parameter~$\alpha\in(0,1]$; that is, the original Ginibre point
process is obtained by taking $\alpha=1$ and it converges weakly to
the homogeneous Poisson point process as $\alpha\to0$.
Indeed, the Ginibre and some other determinantal point processes have
been recognized as a promising class of BS deployment models for
cellular networks due to the observations that they can capture the
spatial characteristics of actual BS deployments (see
\cite{DengZhouHaen15,GomeVassVergMartDecrChen15,LiBaccDhilAndr15}).

A purpose of this tutorial is to provide a brief guide to the Ginibre
and $\alpha$-Ginibre point processes in order for the readers to use
them as a tool for analyzing the performance of cellular networks and
challenging themselves to various new problems arising in modern
cellular networks.
On this account, after reviewing some fundamental and useful
properties of these spatial point processes, we show some existing
results on the performance analysis of cellular network models with
$\alpha$-Ginibre deployed BSs.
For comparison, we mention the results on the related Poisson deployed
BS models as well.

The organization of the paper is as follows.
In the next section, we provide a general spatial stochastic model of
downlink cellular networks and give a few examples.
The signal-to-interference-plus-noise ratio (SINR)---that is a key
quantity for the connectivity in wireless networks---for a typical
user is defined there.
In Section~\ref{sec:GPP}, we introduce the $\alpha$-Ginibre point
process as a model of the BS deployments in cellular networks, where
we first give its definition and then review its fundamental and
useful properties.
In Section~\ref{sec:Analysis}, we show some existing results on the
coverage analysis of cellular networks with $\alpha$-Ginibre deployed
BSs; that is, we give numerically computable forms of coverage
probability---the probability that the SINR for the typical user
achieves a target threshold---for the example models taken in
Section~\ref{sec:Model}.
We finally suggest a problem and the promising direction for future
development in Section~\ref{sec:cncl}.

\section{Spatial model of downlink cellular networks}\label{sec:Model}

We first define a general spatial stochastic model of downlink
cellular wireless networks and then give two examples; one is the most
basic model of a homogeneous single-antenna network and the other is
of a heterogeneous multi-tier multi-antenna network.

Let $\Phi$ denote a point process on $\R^2$ and $X_i$, $i\in\N$,
denote the points of $\Phi$, where the order of $X_1,X_2,\ldots$ is
arbitrary.
Each point~$X_i$, $i\in\N$, represents the location of a BS in a
cellular network and we refer to the BS located at $X_i$ as BS~$i$.
Assuming that the point process $\Phi$ is simple almost surely (a.s.)
and stationary with positive and finite intensity, we focus on a
typical user located at the origin~$o=(0,0)$.
The transmission power of signal from BS~$i$, $i\in\N$, is denoted by
$P_i$.
The random propagation effect of fading and shadowing on the signal
from BS~$i$ to the typical user is denoted by $H_i$, $i\in\N$, when
the BS~$i$ works as a transmitter to the typical user while it is
denoted by $G_i$ when the BS~$i$ works as an interferer for the
typical user, where $H_i$ and $G_i$, $i\in\N$, are nonnegative random
variables.
The path-loss function representing the attenuation of signals with
distance from BS~$i$ is given by $L_i(r)$, $r>0$, where each $L_i$ is
a randomly chosen nonincreasing function on $(0,\infty)$.
Our network model is then described as the stationary marked point
process $\Tilde{\Phi} = \{(X_i, (P_i, H_i, G_i, L_i))\}_{i\in\N}$.

The downlink SINR for the typical user at the origin is defined by
\begin{equation}\label{eq:SINR}
  \SINR_o
  = \frac{S_o(\eta(o))}{I_o(\eta(o)) + w_o},
\end{equation}
where $\eta(x)$ denotes the index of the BS associated with the user
located at $x\in\R^2$ and is determined by a certain association rule
(see, e.g., Examples~\ref{ex1} and~\ref{ex2} below), $S_o(i) =
P_i\,H_i\,L_i(|X_i|)$, $i\in\N$, denotes the desired signal power when
the typical user is served by the BS~$i$ and $I_o(i)$ denotes the
cumulative interference power from all the BSs except BS~$i$ received
by the typical user; that is,
\begin{equation}\label{eq:interference}
  I_o(i) = \sum_{j\in\N\setminus\{i\}} P_j\,G_j\,L_j(|X_j|).
\end{equation}
Also, $w_o$ in \eqref{eq:SINR} denotes a nonnegative constant
representing the noise power at the origin.

\begin{example}[Homogeneous single-antenna network]\label{ex1}
\mbox{}  
The most simple and basic model is that of the homogeneous
single-antenna network, where all the BSs have the same level of
transmission power denoted by a constant~$p$ (i.e.\ $P_i=p$, $i\in\N$).
The propagation effects~$(H_i, G_i)$, $i\in\N$, are independent and
identically distributed (i.i.d.), and also independent of
$\Phi=\{X_i\}_{i\in\N}$.
We often assume the Rayleigh fading and ignore the shadowing for
$\{H_i\}_{i\in\N}$; that is, each $H_i$ is an exponentially
distributed random variable with unit mean, denoted by
$H_i\sim\mathrm{Exp}(1)$.
The path-loss function is also common to all the BSs such that
$L_i(r)=\ell(r)$, which we have in mind is, for example,
$\ell(r)=r^{-2\,\beta}$ or $\ell(r) = \min(1, r^{-2\,\beta})$ with
$\beta>1$.
Each user is served by the nearest BS; that is
$\{\eta(x)=i\}=\{|x-X_i|\le |x-X_j|, j\in\N\}$ for $x\in\R^2$.
Due to the homogeneity of the BSs, the nearest BS association is now
equivalent to the maximum average received power association since
$\Exp(S_o(i)\mid X_i)=p\,\Exp(H_i)\ell(|X_i|)$, where $\Exp(H_i)$ is
identical for all $i\in\N$ and $\ell$ is nonincreasing.
\end{example}

\begin{example}[Multi-tier multi-antenna network]\label{ex2}
Let $K$ denote a positive integer and $\K = \{1,2, \ldots$, $K\}$.
Each BS is classified into one of $K$ distinct tiers (classes) and a
BS of tier~$k\in\K$ has the specific transmission power~$p_k$, the
number of antennas~$m_k$, the number of users to be
served~$\psi_k$~($\le m_k$) and the path-loss function~$\ell_k(r)$.
This model represents the multi-input multi-output (MIMO)
transmission in a heterogeneous network (HetNet).
Assuming the Rayleigh fading on all links and the single receiving
antenna for each user, the discussion in \cite{HuanPapaVenk12} (see,
e.g, \cite{DhilKounAndr13} also) enables us to suppose that the
channel power distributions of both the associated and interfering
links follow the Erlang distributions with different shape parameters;
that is, when the BS~$i$ is of tier~$k$, $H_i\sim\Gam(\delta_k, 1)$
with $\delta_k = m_k - \psi_k+1$ and $G_i\sim\Gam(\psi_k,1)$, where
``$\Gam$'' denotes the Gamma distribution.
Let $\xi_i$ denote the tier of BS~$i$.
This model is then described as the marked point
process~$\Phi_\xi = \{(X_i,\xi_i)\}_{i\in\N}$ since $(P_i,
L_i)=(p_{\xi_i}, \ell_{\xi_i})$, and $(H_i, G_i)$, $i\in\N$, are
conditionally mutually independent given $\xi_i$, $i\in\N$.
As for the BS association, we introduce another parameter~$b_k>0$,
$k\in\K$, called the bias factor, and adopt the flexible cell
association rule (see \cite{JoSangXiaAndr12,GuptDhilVishAndr14}); that
is, each user is served by the BS that supplies the maximum
biased-average-received-power;
\[
  \{\eta(x)=i\}
  = \bigl\{
      b_{\xi_i} p_{\xi_i} \delta_{\xi_i} \ell_{\xi_i}(|x-X_i|)
       \ge b_{\xi_j} p_{\xi_j} \delta_{\xi_j}\ell_{\xi_j}(|x-X_j|),\:
          j\in\N
     \bigr\},
\]
where $p_k\delta_k\ell_k(|X_i|) = \Exp(S_o(i)\mid X_i, \xi_i=k)$
represents the average received signal power for the typical user from
the BS~$i$ when this BS is of tier~$k$.
\end{example}

\section{$\alpha$-Ginibre point processes and their
  properties}\label{sec:GPP}

In this section, we give a brief introduction to the Ginibre and
$\alpha$-Ginibre point processes.
Since these point processes belong to a class of the determinantal
point processes on the complex plane~$\C\simeq\R^2$, we first define a
general determinantal point process on $\R^d$.
Readers are referred to
e.g.\ \cite{Sosh00,ShirTaka03,HougKrisPereVira09,LavaMollRuba15} for
further details.

\subsection{Determinantal point processes}

Let $\Phi$ denote a simple point process on $\R^d$ and let
$\rho^{(n)}$ denote its $n$th product density functions (joint
intensities) with respect to a locally finite measure $\nu$ on $\R^d$;
that is, for any symmetric and continuous function $f$ with bounded
support on $\R^{d\times n}$,
\begin{align}\label{eq:productdensity}
  \Exp\Biggl[
    \sum_{\substack{X_1,\ldots,X_n\in\Phi\\ X_i\ne X_j, i\ne j}}
       f(X_1,X_2,\ldots,X_n)
  \Biggr]
  = \int\!\!\int\!\cdots\!\int_{\R^{d\times n}}
      f(x_1,x_2,\ldots,x_n)
       \rho^{(n)}(x_1,x_2,\ldots,x_n)\,
     \prod_{i=1}^n\nu(\dd x_i).
\end{align}
The point process $\Phi$ is then said to be a determinantal point
process on $\R^d$ with kernel~$K$:~$\R^d\times\R^d\to\C$ with respect
to the reference measure $\nu$ if the product density
function~$\rho^{(n)}$ in \eqref{eq:productdensity} is given by
\begin{equation}\label{eq:determinantal}
  \rho^{(n)}(x_1,x_2,\ldots,x_n)
  = \det\bigl(K(x_i,x_j)\bigr)_{i,j=1}^n,
\end{equation}
where ``$\det$'' denotes the determinant.
In order for the point process~$\Phi$ to be well-defined, we usually
assume that (i) the kernel~$K$ is continuous on $\R^d\times\R^d$\!,
(ii)~$K$ is Hermitian in the sense that $K(x, y) = \Bar{K(y, x)}$ for
$x$, $y\in\R^d$, where $\Bar{z}$ denotes the complex conjugate of
$z\in\C$ and (iii)~the integral operator on $L^2(\R^d)$ corresponding
to $K$ is of locally trace class with the spectrum in $[0,1]$;
that is, $0\le (Kf, f)\le (f,f)$ for any $f \in L^2(\mathbb{R}^d)$,
where the inner product is given by $(f,g) = \int_{\R^d}
f(x)\,\Bar{g(x)}\,\nu(\dd x)$, and for any bounded set $C\in\B(\R^d)$,
the restriction $K_C$ of $K$ on $C$ has eigenvalues~$\kappa_{C,i}$,
$i\in\N$, satisfying $\sum_{i\in\N}\kappa_{C,i}<\infty$.
Under these conditions, $\kappa_{C, i}\in[0,1]$ holds for any bounded
$C\in\B(\R^d)$ and $i\in\N$ (see, e.g.,
\cite[Chap.~4]{HougKrisPereVira09}).
Then the number of points of $\Phi$ falling in $C$ has the
distribution of the sum of independent Bernoulli random
variables~$B_{C,i}$ with $\Prb(B_{C,i}=1)=\kappa_{C,i}$, $i\in\N$;
that is,
\begin{equation}\label{eq:DPP0}
  \Phi(C) \eqd \sum_{i\in\N}B_{C,i},
\end{equation}
where ``$\eqd$'' denotes equality in distribution.
This immediately leads to the expectation and variance of $\Phi(C)$;
\[
  \Exp\Phi(C) = \sum_{i\in\N}\kappa_{C,i},
  \quad
  \Var\,\Phi(C) = \sum_{i\in\N}\kappa_{C,i}\,(1-\kappa_{C,i}),
\]
where it should be noted that $\Var\,\Phi(C)\le\Exp\Phi(C)<\infty$ for
any bounded $C\in\B(\R^d)$.

The Palm distribution is a basic concept in the point process theory
and formalizes the notion of the conditional distribution of a point
process given that it has a point at a specific location.
The following proposition states that a determinantal point process is
closed under the operation of taking the reduced Palm
distribution\footnote{%
The reduced Palm distribution formalizes the notion of the conditional
distribution of a point process given that the process has a point at
a specific location but excluding this point on which the process is
conditioned.}.

\begin{proposition}[\cite{ShirTaka03}]
Let $\Phi$ denote a determinantal point process on $\R^d$ with
kernel~$K$ with respect to the reference measure~$\nu$.
Then, for almost every $x_0\in\R^d$ with respect to the
measure~$K(x,x)\,\nu(\dd x)$, $\Phi$ is also determinantal under the
reduced Palm distribution given a point at $x_0$ and the corresponding
kernel~$K^{x_0}$ is given by
\begin{equation}\label{eq:Palm}
  K^{x_0}(x,y)
  = \frac{K(x,y)\,K(x_0,x_0)-K(x,x_0)\,K(x_0,y)}
         {K(x_0,x_0)},
\end{equation}
whenever $K(x_0,x_0)>0$.
\end{proposition}

\subsection{$\alpha$-Ginibre point processes}

For $\alpha\in(0,1]$, a determinantal point process~$\Phi_\alpha$ on
$\C$~($\simeq\R^2$) is said to be an $\alpha$-Ginibre point process
when its kernel~$K_\alpha$ on $\C\times\C$ is given by
\begin{equation}\label{eq:alpha-kernel}
  K_\alpha(z,w)=\ee^{z\Bar{w}/\alpha},\quad z, w\in\C,
\end{equation}
with respect to the modified Gaussian measure
\begin{equation}\label{eq:alpha-measure}
  \nu_\alpha(\dd z) = \frac{1}{\pi}\,e^{-|z|^2/\alpha}\,\mu(\dd z),
\end{equation}
where $\mu$ denotes the Lebesgue measure on $(\C,\B(\C))$.
The choice of pair~$(K_\alpha, \nu_\alpha)$ is not unique and the
kernel~$\Tilde{K}_\alpha(z,w) =
\pi^{-1}\,e^{-(|z|^2+|w|^2)/(2\alpha)}\,e^{z\Bar{w}/\alpha}$ with
respect to the Lebesgue measure~$\mu$ defines the same process as
$\Phi_\alpha$.
The process with $\alpha=1$ gives the original Ginibre point process.

Let $\Tilde{\rho}_\alpha^{(n)}$, $n\in\N$, denote the product density
functions of $\Phi_\alpha$ with respect to the Lebesgue measure. 
For example, the first two product densities are then given by
\eqref{eq:determinantal} as
\begin{align}
  \Tilde{\rho}_\alpha^{(1)}(z)
  &= \Tilde{K}_\alpha(z,z) = \pi^{-1},
  \label{eq:alpha-intensity}\\
  \Tilde{\rho}_\alpha^{(2)}(z,w)
  &= \frac{1-\ee^{-|z-w|^2/\alpha}}{\pi^2}.
\end{align}
Note that both the product densities are motion invariant (invariant
under translation and rotation).
In fact, one can show that the $n$th product density is motion
invariant for each $n\in\N$, and hence the $\alpha$-Ginibre point
process is motion invariant; that is, stationary and isotropic.
We further see that $\Tilde{\rho}_\alpha^{(2)}(z,w)\to \pi^{-2}$ as
$\alpha\to0$, converging to the second-order product density of the
homogeneous Poisson point process with intensity $\pi^{-1}$.
Again, one can show that $\Phi_\alpha$ converges weakly to the
homogeneous Poisson point process with intensity $\pi^{-1}$ as
$\alpha\to0$ (see \cite{Gold10}).
This suggests that the $\alpha$-Ginibre point process constitutes an
intermediate class between the original Ginibre and homogeneous
Poisson point processes by the parameter~$\alpha\in(0,1]$.
Figure~\ref{fig:samples} shows samples of the Poisson and
$\alpha$-Ginibre point processes with the same intensity.
We can see that the configuration of the points becomes more regular
as the value of $\alpha$ becomes larger.
  
\begin{figure*}[t]%
\begin{center}
\hbox to\hsize{\kern-3em%
\includegraphics[width=.45\hsize]{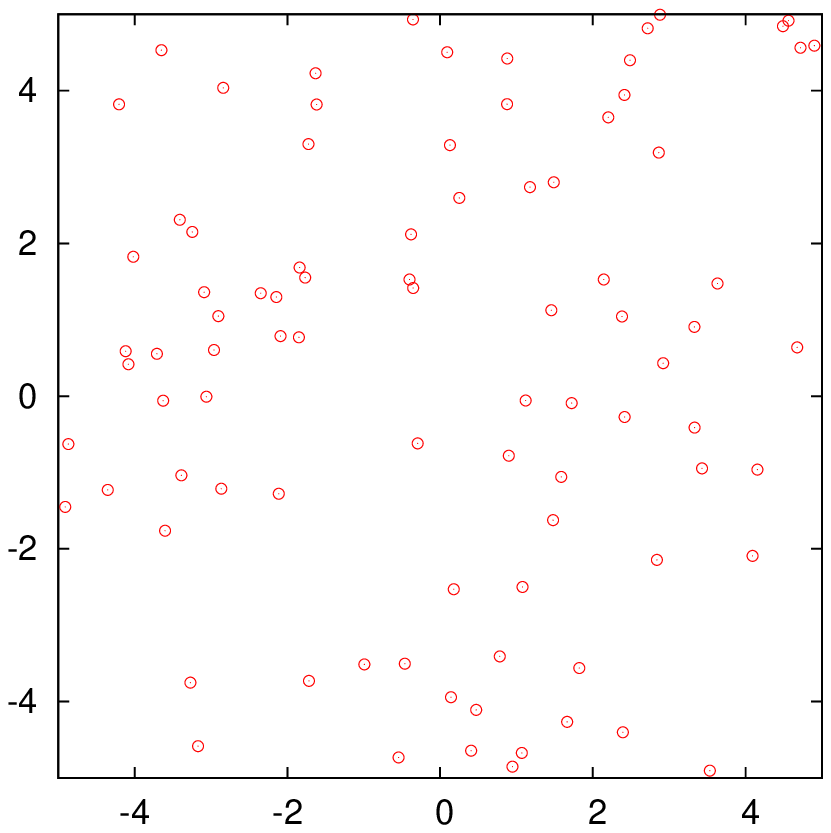}\hss
\includegraphics[width=.45\hsize]{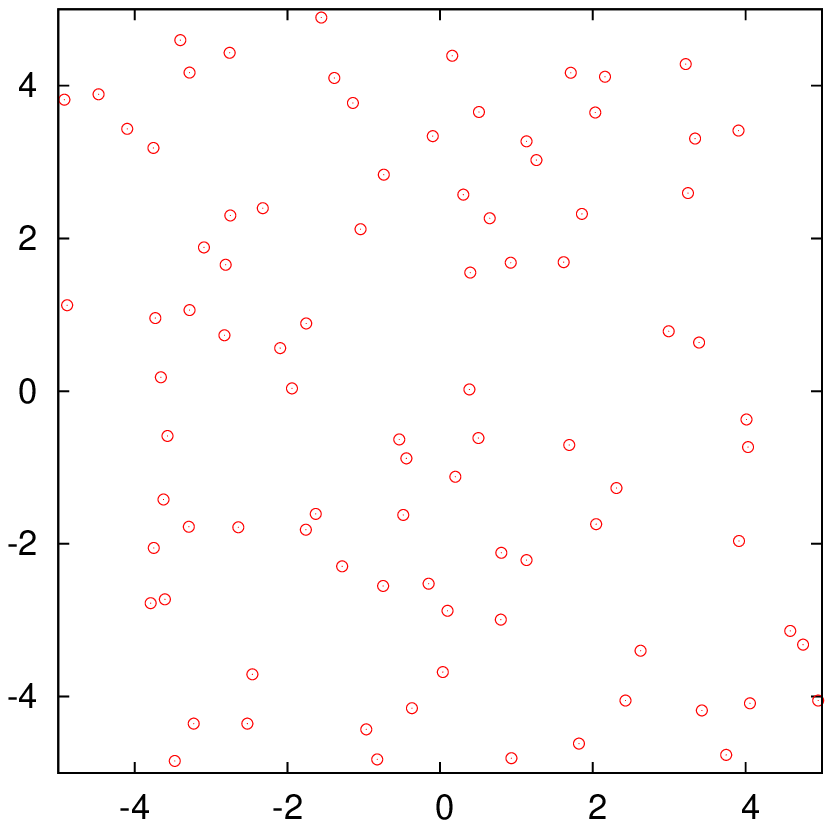}\hss
\includegraphics[width=.45\hsize]{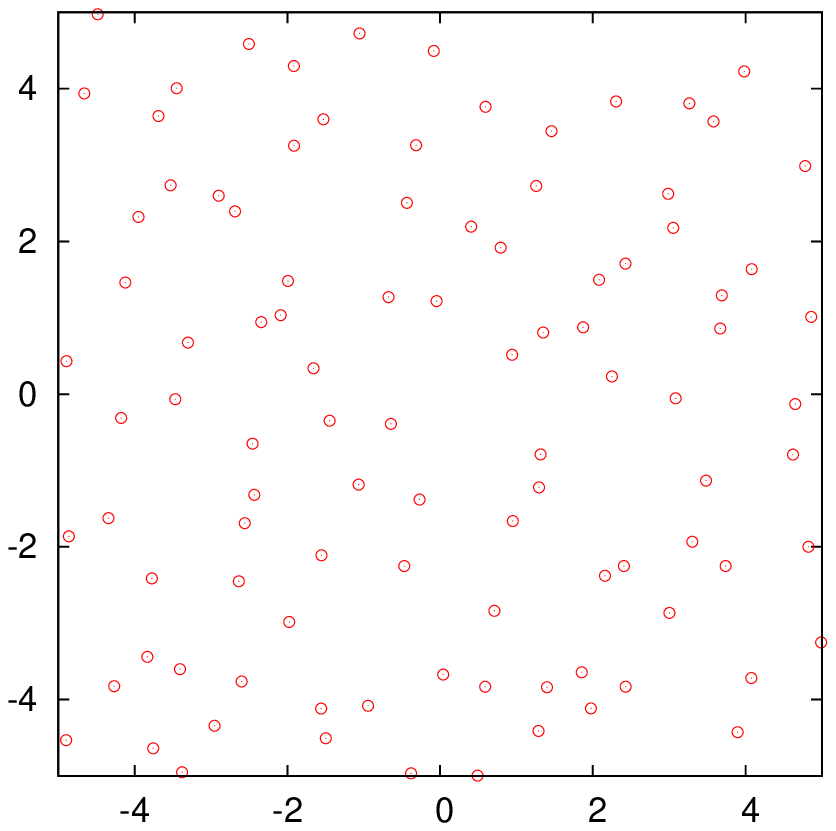}%
\kern-2.5em}%
\end{center}%
\caption{Samples of the Poisson~($\alpha\to0$, left), $\alpha$-Ginibre
  ($\alpha=0.5$, center) and original Ginibre ($\alpha=1$, right)
  point processes with the same intensity.~\cite{NakaMiyo14}}%
  \label{fig:samples}%
\end{figure*}

\begin{remark}\label{rmk:scale}
As seen in~\eqref{eq:alpha-intensity}, the $\alpha$-Ginibre point
process has the intensity~$\pi^{-1}$ with respect to the Lebesgue
measure; that is, for $C\in\B(\C)$,
\[
  \Exp\Phi(C)
  = \int_C \rho_\alpha^{(1)}(z)\,\nu_\alpha(\dd z)
  = \frac{\mu(C)}{\pi}.
\]
However, we can consider the process with an arbitrary fixed
intensity~$\lambda\in(0,\infty)$ by scaling.
The kernel and reference measure of the scaled $\alpha$-Ginibre point
process with intensity~$\lambda$ are respectively given by
$K_{\alpha,\lambda}(z,w) = \ee^{\pi\lambda z\Bar{w}/\alpha}$ and
$\nu_{\alpha,\lambda}(\dd z) =
\lambda\,\ee^{-\pi\lambda|z|^2/\alpha}\,\mu(\dd z)$.
Or equivalently, the kernel~$\Tilde{K}_\alpha^{(\lambda)}(z,w) =
\lambda\,\ee^{-\pi\lambda(|z|^2+|w|^2)/(2\alpha)}\,\ee^{\pi\lambda
  z\Bar{w}/\alpha}$ with respect to the Lebesgue measure $\mu$ defines
the same process.
\end{remark}

We next see the nonzero eigenvalues and the corresponding
eigenfunctions of the integral operator corresponding to the
kernel~$K_\alpha$.
Let
\begin{equation}\label{eq:eigenfunction}
  \phi_{\alpha,i}(z)
  = \frac{z^{i-1}}{\sqrt{(i-1)!\,\alpha^i}},
  \quad i\in\N.
\end{equation}
Then we can check that $\phi_{\alpha,i}$, $i\in\N$, are the
orthonormal eigenfunctions of $K_\alpha$ corresponding to the
eigenvalue~$\alpha$ satisfying
\[
  \int_{\C}
    \phi_{\alpha,i(}z)\,\Bar{\phi_{\alpha,j}(z)}\,
  \nu_\alpha(\dd z)
  = \Biggl\{
      \vcenter{%
      \hbox{$1$\quad for $i=j$,}
      \hbox{$0$\quad for $i\ne j$.}
      }%
\]
Thus, Mercer's spectral expansion~(\cite{Merc09}) holds such that
\[
  K_\alpha(z,w)
  = \sum_{i=1}^\infty
      \alpha\,\phi_{\alpha,i}(z)\,\Bar{\phi_{\alpha,i}(w)},
  \quad z, w\in\C.
\]
Now, let $D_r$ denote the disk on $\C$ centered at the origin with
radius~$r$.
Then $\phi_{\alpha,i}$, $i\in\N$, in \eqref{eq:eigenfunction} are also
orthogonal eigenfunctions (but not normal now) of the
restriction~$K_{\alpha,D_r}$ of $K_\alpha$ on $D_r$ corresponding to
the eigenvalues
\begin{equation}\label{eq:alpha_eigenvalue}
  \kappa_{\alpha,D_r,i}
  = \alpha\,P(i,r^2/\alpha)
  = \alpha\,\frac{\gamma(i,r^2/\alpha)}{\Gamma(i)},
  \quad i\in\N,
\end{equation}
where $P(x,y)= \gamma(x,y)/\Gamma(x)$ denotes the regularized lower
incomplete Gamma function with the lower incomplete Gamma
function~$\gamma(x,y)=\int_0^y t^{x-1}\,\ee^{-t}\,\dd t$ and the usual
Gamma function~$\Gamma(x)=\gamma(x,\infty)$.
Let $\chi_i$, $i\in\N$, denote i.i.d.\ Bernoulli random variables with
$\Prb(\chi_i=1)=\alpha$ and let $Y_i$, $i\in\N$, denote mutually
independent random variables with $Y_i\sim\Gam(i,\alpha^{-1})$, where
$\{\chi_i\}_{i\in\N}$ and $\{Y_i\}_{i\in\N}$ are also independent of
each other.
Then, since $\Prb(Y_i\le r^2) = P(i,r^2/\alpha)$, \eqref{eq:DPP0} and
\eqref{eq:alpha_eigenvalue} imply
\[
  \Phi_\alpha(D_r)
  \eqd \sum_{i\in\N} \chi_i\,\ind{\{Y_i\le r^2\}}.
\]
This observation is closely related to the following proposition,
which is a generalization of Kostlan's result~\cite{Kost92} for the
original Ginibre point process (see also
\cite[Theorem~4.7.1]{HougKrisPereVira09}).

\begin{proposition}\label{prp:Kostlan}
Let $X_i$, $i\in\N$, denote the points of the $\alpha$-Ginibre point
process.
Then, the set $\{|X_i|^2\}_{i\in\N}$ has the same distribution as
$\check{\boldsymbol{Y}} = \{\check{Y}_i\}_{i\in\N}$, which is
extracted from $\boldsymbol{Y}=\{Y_i\}_{i\in\N}$ such that $Y_i$,
$i\in\N$, are mutually independent with $Y_i\sim\Gam(i,\alpha^{-1})$
and each $Y_i$ is added in $\check{\boldsymbol{Y}}$ with probability
$\alpha$ and discarded with $1-\alpha$ independently of others. 
\end{proposition}

Indeed, the $\alpha$-Ginibre point process~$\Phi_\alpha$ is obtained
from the original Ginibre process $\Phi=\Phi_1$ by retaining each
point of $\Phi$ with probability~$\alpha$ (removing it with
$1-\alpha$) independently, and then applying the homothety of
ratio~$\sqrt{\alpha}$ to the retained points in order to maintain the
original intensity of the Ginibre process~$\Phi$ (\cite{Gold10}).
Proposition~\ref{prp:Kostlan} is useful for analyzing the cellular
network models described in the preceding section since the path-loss
function usually depends only on the distance from a BS.
When we consider the scaled $\alpha$-Ginibre point process with
intensity~$\lambda\in(0,\infty)$ as in Remark~\ref{rmk:scale},
$\Gam(i,\alpha^{-1})$ in the above proposition is replaced by
$\Gam(i,\pi\lambda/\alpha)$.

We can extend Proposition~\ref{prp:Kostlan} to the process under the
Palm distribution.
Applying \eqref{eq:Palm} to \eqref{eq:alpha-kernel}, the kernel
$K_\alpha^o$ of the $\alpha$-Ginibre point process under the reduced
Palm distribution given a point at the origin is
\begin{equation}\label{eq:Palm_Kernel}
  K_\alpha^o (z,w) = \ee^{z\overline{w}/\alpha} - 1,
\end{equation}
with respect to the same reference measure $\nu_\alpha$ in
\eqref{eq:alpha-measure}.
Thus, the first product density is given by
\begin{align*}
  \rho_\alpha^{o (1)}(z)\,\nu_\alpha(\dd z)
  = K_\alpha^o (z,z)\,\nu_\alpha(\dd z)
  = \frac{1}{\pi}\,(1-\ee^{-|z|^2/\alpha})\,\mu(\dd z).
\end{align*}
Note that the $\alpha$-Ginibre point process is no longer stationary
under the Palm distribution and the intensity function is increasing
according to the distance from the origin.
The following Proposition is obtained by applying the kernel
\eqref{eq:Palm_Kernel} to Theorem 4.7.1 of \cite{HougKrisPereVira09}.

\begin{proposition}\label{prp:p_kostlan}
Let $X^o_i$, $i\in\N$, denote the points of the $\alpha$-Ginibre point
process under the reduced Palm distribution.
Then, the set $\{|X^o_i|^2\}_{i \in \N}$ has the same distribution as
$\check{\boldsymbol{Y}}^o = \{\check{Y}^o_i\}_{i\in\N}$, which is
extracted from $\boldsymbol{Y}=\{Y^o_i\}_{i\in\N}$ such that $Y^o_i$,
$i\in\N$, are mutually independent with
$Y^o_i\sim\Gam(i+1,\alpha^{-1})$ and each $Y^o_i$ is added in
$\check{\boldsymbol{Y}}^o$ with probability $\alpha$ and discarded
with $1-\alpha$ independently of others.
\end{proposition}

Note that $\check{\boldsymbol{Y}}^o$ in
Proposition~\ref{prp:p_kostlan} is obtained from
$\check{\boldsymbol{Y}}$ in Proposition~\ref{prp:Kostlan} by removing
the exponentially distributed random
variable~$Y_1\sim\Gam(1,\alpha^{-1})$ if it is retained with
probability~$\alpha$ (see \cite{Gold10}).

\section{Coverage analysis}\label{sec:Analysis}

In this section, we show some existing results on the coverage
analysis of the cellular network models described in
Section~\ref{sec:Model}; that is, we give the numerically computable
forms of the coverage probability for the two examples in
Section~\ref{sec:Model}.
Here, the coverage probability is defined as the tail probability
$\Prb(\SINR_o>\theta)$, $\theta>0$, of the SINR in \eqref{eq:SINR},
which represents the probability that the SINR for the typical user
achieves a target threshold~$\theta$.

\subsection{Homogeneous single-antenna network}

We here derive a numerically computable form of the coverage
probability for the homogeneous single-antenna network model in
Example~\ref{ex1}, where the BSs are deployed according to the
$\alpha$-Ginibre point process with intensity~$\lambda\in(0,\infty)$.
The corresponding result for the model with Poisson deployed BSs is
also derived.
The proof for the Poisson deployed BS model mainly follows
\cite{AndrBaccGant11} while that for the $\alpha$-Ginibre deployed BS
model does \cite{MiyoShir14a,NakaMiyo14}.

\begin{theorem}[\cite{AndrBaccGant11,MiyoShir14a,NakaMiyo14}]\label{thm:coverage}
Consider the homogeneous single-antenna cellular network
model in Example~\ref{ex1} with the path-loss
function~$\ell(r)=r^{-2\beta}$, $r>0$, for $\beta>1$, where
$H_i\sim\mathrm{Exp}(1)$, $i\in\N$, (Rayleigh fading) and $G_i$,
$i\in\N$, are i.i.d.
When the point process~$\Phi$ is the homogeneous Poisson point process
with intensity~$\lambda\in(0,\infty)$, the coverage probability for
the typical user is given by
\begin{align}\label{eq:coverage_PPP}
  \Prb(\SINR_o^{(\text{PPP})}>\theta)
  = \int_0^\infty
      \exp\biggl\{
        -\frac{\theta\,w_o}{p}\,
         \biggl(\frac{t}{\pi\,\lambda}\biggr)^\beta
        -t\,\Bigl(1 + \tau(\theta,\beta)\Bigr)
      \biggr\}\,
    \dd t,
\end{align}
where
\begin{equation}\label{eq:tau}
  \tau(\theta,\beta)
  = \frac{\theta^{1/\beta}}{\beta}
    \int_{1/\theta}^\infty
      \bigl(1 - \Lpl_G(u^{-1})\bigr)\,
      u^{-1+1/\beta}\,
    \dd u,
\end{equation}
and $\Lpl_G$ denotes the Laplace transform of $G_i$, $i\in\N$.
On the other hand, when $\Phi$ is the $\alpha$-Ginibre point process
with intensity~$\lambda$,
\begin{align}\label{eq:coverage_GPP}
  \Prb(\SINR_o^{(\text{$\alpha$-GPP})}>\theta)
  = \alpha
    \int_0^\infty
      \exp\biggl\{
        -t -\frac{\theta\,w_o}{p}\,
            \biggl(\frac{\alpha\,t}{\pi\lambda}\biggr)^\beta
      \biggr\}\,
       M_\alpha(t, \theta, \beta)\,S_\alpha(t, \theta, \beta)\,
     \dd t,
\end{align}
where
\begin{align}
  M_\alpha(t,\theta,\beta)
  &= \prod_{i=0}^\infty
       \Bigl[
         1 - \alpha + \alpha\,J_i(t,\theta,\beta)
       \Bigr],
  \label{eq:M}\\
  S_\alpha(t,\theta,\beta)
  &= \sum_{i=0}^\infty
       \frac{t^i}{i!}\,
       \Bigl[
         1 - \alpha + \alpha\,J_i(t,\theta,\beta)
       \Bigr]^{-1},
  \label{eq:S}
\end{align}
with
\begin{equation}\label{eq:J}
  J_i(t,\theta,\beta)
  = \frac{1}{i!}
    \int_t^\infty
      \ee^{-u}\,u^i\,
      \Lpl_G\biggl(\theta\,\biggl(\frac{t}{u}\biggr)^\beta\biggr)\,
    \dd u.
\end{equation}
\end{theorem}

For the proof of \eqref{eq:coverage_PPP}--\eqref{eq:tau} for the
Poisson deployed BS model, we use the probability generating
functional for point processes.

\begin{definition}
Let $\Phi=\{X_i\}_{i\in\N}$ denote a point process on $\R^d$ with
intensity measure~$\Lambda$; that is, $\Exp\Phi(C)=\Lambda(C)$ for
$C\in\B(\R^d)$.
For any measurable function~$v$:~$\R^d\to[0,1]$ such that
$\int_{\R^d}\bigl(1-v(x)\bigr)\,\Lambda(\dd x) < \infty$, the
probability generating functional of the point process~$\Phi$ is
defined as
\[
  \GF_\Phi(v)
  = \Exp\biggl[
      \prod_{i\in\N} v(X_i)
    \biggr].
\]
\end{definition}

\begin{proposition}[e.g., {\cite[Sec.~9.4]{DaleVere08}}]
For the Poisson point process~$\Phi$ on $\R^d$ with intensity
measure~$\Lambda$, its probability generating functional is given as
\begin{equation}\label{eq:Poisson_GF}
  \GF_\Phi^{(\text{PPP})}(v)
  = \exp\biggl\{
      - \int_{\R^d}\bigl(1-v(x)\bigr)\,\Lambda(\dd x)
    \biggr\}.
\end{equation}
\end{proposition}

Note that, if $\Phi$ is stationary with intensity~$\lambda$, then
$\Lambda(\dd x)$ above is replaced by $\lambda\,\dd x$.

\begin{proof}[Proof of Theorem~\ref{thm:coverage}:]
In the definition of the SINR in~\eqref{eq:SINR}, each $H_i$ is
independent of $\Phi=\{X_i\}_{i\in\N}$ and
$\{G_j\}_{j\in\N\setminus\{i\}}$.
Also, $\eta(o)$ is determined by $\Phi=\{X_i\}_{i\in\N}$.
Thus, conditioning on $\Phi=\{X_i\}_{i\in\N}$ and
$\{G_j\}_{j\in\N\setminus\{\eta(o)\}}$, and using
$H_i\sim\mathrm{Exp}(1)$, $i\in\N$, we have
\begin{align*}
  \Prb(\SINR_o>\theta)
  = \Prb\biggl(
      H_{\eta(o)} > \theta\,\frac{I_o(\eta(o))+w_o}{p\,\ell(|X_{\eta(o)}|)}
    \biggr)
  = \Exp\biggl[
      \exp\biggl\{
        -\theta\,\frac{I_o(\eta(o))+w_o}{p\,\ell(|X_{\eta(o)}|)}
      \biggr\}
    \biggr].
\end{align*}
Furthermore, the definition of the
interference~\eqref{eq:interference} and the Laplace transform of
$G_j$, $j\in\N$, lead to
\begin{align}\label{eq:cov0}
  \Prb(\SINR_o>\theta)
   = \Exp\biggl[
       \exp\biggl\{
         -\frac{\theta\,w_o}{p\,\ell(|X_{\eta(o)}|)}
       \biggr\}
       \prod_{j\in\N\setminus\{\eta(o)\}}\!
         \Lpl_G\biggl(
           \theta\,\frac{\ell(|X_j|)}{\ell(|X_{\eta(o)}|)}
         \biggr)
     \biggr],
\end{align}
which is the starting point for both the Poisson and $\alpha$-Ginibre
deployed BS cellular network models.

We first show \eqref{eq:coverage_PPP}--\eqref{eq:tau} for the Poisson
deployed BS model.
For the homogeneous Poisson point process~$\Phi$ on $\R^2$ with
intensity~$\lambda$, the distribution for the distance to the nearest
point from the origin is given by
\begin{equation}\label{eq:Poisson_nearest}
  \Prb(|X_{\eta(o)}|>r)
  = \Prb\bigl(\Phi(D_r)=0\bigr)
  = \ee^{-\lambda \pi r^2},
\end{equation}
where $D_r$ denotes the disk centered at the origin with radius~$r$.
Given $|X_{\eta(o)}|=r$, other points of $\Phi$ also follow the
Poisson point process on $\R^2\setminus D_r$, and thus applying the
probability generating functional~\eqref{eq:Poisson_GF}, we obtain
\begin{align}\label{eq:Poisson_cov0}
  \Exp\biggl[
    \prod_{j\in\N}
      \Lpl_G\biggl(
        \theta\,\frac{\ell(|X_j|)}{\ell(r)}
      \biggr)
  \biggm|
    X_j \in \R^2\!\setminus\! D_r,\: j\in\N
  \biggr]
  &= \exp\biggl\{
       -\lambda
        \int_{|x|>r}
          \biggl[
            1 - \Lpl_G\biggl(\theta\,\frac{\ell(|x|)}{\ell(r)}\biggr)
          \biggr]\,
        \dd x
     \biggr\}
  \nonumber\\
  &= \exp\biggl\{
       -2\pi\lambda
         \int_r^\infty
           \biggl[
             1 - \Lpl_G\biggl(\theta\,\frac{\ell(s)}{\ell(r)}\biggr)
         \biggr]\,s\,
       \dd s
     \biggr\}.
\end{align}
Hence, applying \eqref{eq:Poisson_nearest}, \eqref{eq:Poisson_cov0}
and $\ell(r)=r^{-2\beta}$ to \eqref{eq:cov0} yields
\eqref{eq:coverage_PPP}--\eqref{eq:tau} after some manipulations.

On the other hand, for the $\alpha$-Ginibre deployed BS model, we use
$\boldsymbol{Y}=\{Y_i\}_{i\in\N}$ in Proposition~\ref{prp:Kostlan}
such that $Y_i$, $i\in\N$, are mutually independent and each $Y_i$ is
retained with probability $\alpha$ independently of others.
Thus, dividing the cases in each of which the point corresponding to
$Y_i$ is retained and associated with the typical user,
\eqref{eq:cov0} with $\ell(r)=r^{-2\beta}$ reduces to
\begin{align*}
  \Prb(\SINR_o^{(\text{$\alpha$-GPP})}>\theta)
  &= \alpha
     \sum_{i\in\N}
       \Exp\biggl[
         \exp\biggl\{
           -\frac{\theta\,w_o}{p}\,{Y_i}^\beta
         \biggr\}
         \prod_{j\in\N\setminus\{i\}}
           \biggl\{
             1 - \alpha
             + \alpha\,\Lpl_G\biggl(
                 \theta\,\biggl(\frac{Y_i}{Y_j}\biggr)^\beta
               \biggr)\,
               \ind{\{Y_j\ge Y_i\}}
           \biggr\}
       \biggr].
\end{align*}
Finally, applying $Y_i\sim\Gam(i,\pi\lambda/\alpha)$, $i\in\N$ yields
\eqref{eq:coverage_GPP}--\eqref{eq:J} after some manipulations.
\end{proof}

Figure~\ref{fig:coverage} shows the comparison result of the coverage
probability with different values of $\alpha$.
Each plot indicates the coverage probability for a given value of
$\theta$ in the case of $w_o=0$ (noise-free) and $\beta=2$ (i.e.,
$\ell(r)=r^{-4}$).
It seems that the coverage probability is increasing in $\alpha$.
However, a numerical result in \cite{MiyoShir14b} shows that the
coverage probability is not always monotone in $\alpha$ as
$\theta\to\infty$.

\begin{figure}[t]%
\begin{center}
\includegraphics[width=.6\hsize]{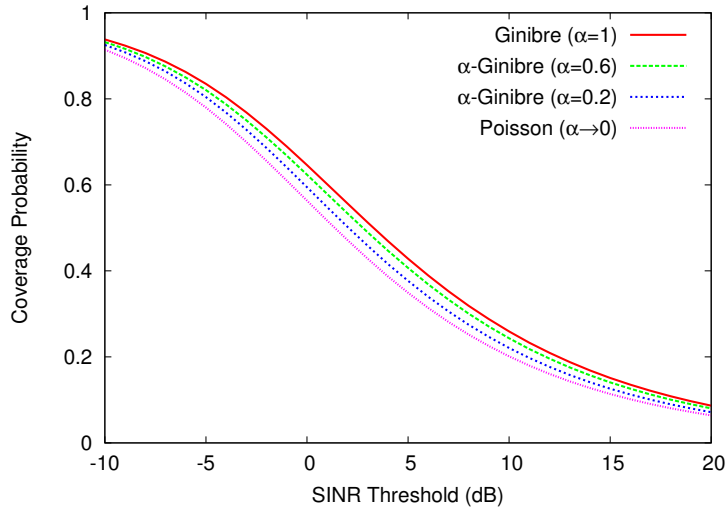}%
\end{center}%
\caption{Comparison of coverage probability in terms of $\alpha$ in
  the single tier model ($\ell(r)=r^{-4}$, no
  noise).~\cite{NakaMiyo14}}%
  \label{fig:coverage}%
\end{figure}

\subsection{Two-tier Ginibre-Poisson overlaid network}

In this subsection, we consider the case of $K=2$ in
Example~\ref{ex2}, where the BSs of tier~$1$ are deployed according to
the $\alpha$-Ginibre point process $\Phi_1$ with intensity~$\lambda_1$
while the BSs of tier~$2$ follow the homogeneous Poisson point
process~$\Phi_2$ with intensity~$\lambda_2$.
This represents that, in heterogeneous multi-tier cellular networks, the
macro BSs are deployed rather systematically while the femto BSs are
located in an opportunistic manner.
We assume that the two point processes~$\Phi_1$ and $\Phi_2$ are
independent of each other.
In the coverage of users, the target thresholds can differ for the two
tiers; that is, the SINR should be larger than $\theta_k$ when a user
is served by a BS of tier~$k$ for $k=1$, $2$. 

For the ease of understanding, we impose some extent of simplifying
setting (see \cite{KobaMiyo16} for a general setting).
First, we ignore the noise power and set $w_o=0$, in this case, the
SINR in \eqref{eq:SINR} is called the signal-to-interference
ratio~(SIR).
Furthermore, we only consider the case where the number of users
served at each BS is equal to the number of antennas; that is,
$m_k=\psi_k$ for $k=1$, $2$.
This case reduces to the single-input single-output (SISO)
transmission when $m_k=\psi_k=1$ while this is called the full form of
space-division multiple access (full SDMA) transmission when
$m_k=\psi_k>1$.
In this setting, $H_i\sim\mathrm{Exp}(1)$, $i\in\N$, since
$\delta_k=1$ for each $k\in\K$, and they are mutually independent.

\begin{theorem}\label{thm:MIMO}
Consider the two-tier multi-antenna cellular network model in
Example~\ref{ex2} with $K=2$, $\ell_1(r)=r^{-2\beta_1}$ and
$\ell_2(r)=r^{-2\beta_2}$.
Then, under the setting described above, the coverage probability for
the typical user is given by
\begin{subequations}
\begin{align}
  \Prb(\SIR_o^{(\mathrm{MIMO})} > \theta_{\xi_{\eta(o)}})
  &= \alpha
     \int_0^\infty
       M_\alpha(t, \theta_1, \beta_1)\,S_\alpha(t, \theta_1,\beta_1)
  \label{eq:MIMO1}\\
  &\qquad\mbox{}\times
       \exp\biggl\{
         -t
         - C_\alpha^{(1,2)}(t)\,
           \bigl(1 + \tau_{1,2}(\theta_1,\beta_2\bigr)
       \biggr\}\,
     \dd t
  \label{eq:MIMO2}\\
  &\quad\mbox{}
   + \int_0^\infty
       M_\alpha^{(2,1)}(t,\theta_2,\beta_1,\beta_2)
  \label{eq:MIMO3}\\
  &\quad\qquad\mbox{}\times
       \exp\Bigl\{
         -t\,\bigl( 1 + \tau(\theta_2,\beta_2) \bigr)
       \Bigr\}\
     \dd t
  \label{eq:MIMO4}
\end{align}
\end{subequations}
where $M_\alpha(t,\theta,\beta)$ and $S_\alpha(t,\theta,\beta)$ are
the same as in \eqref{eq:M} and \eqref{eq:S} respectively with
$\Lpl_G(s) = \Lpl_{G,1}(s)=(1+s)^{-\psi_1}$ in $J_i$ in \eqref{eq:J}.
Moreover,
\begin{align*}
  C_\alpha^{(1,2)}(t)
  &= \pi\lambda_2\,
     \biggl(\frac{b_2p_2}{b_1p_1}\biggr)^{1/\beta_2}\,
     \biggl(\frac{\alpha t}{\pi\lambda_1}\biggr)^{\beta_1/\beta_2},
  \\
  \tau_{1,2}(\theta,\beta)
  &= \frac{\theta^{1/\beta}}{\beta}
    \int_{1/\theta}^\infty
      \biggl[
        1 - \biggl(\frac{u}{u+b_1/b_2}\biggr)^{\psi_2}
      \biggr]\,
      u^{-1+1/\beta}\,
    \dd u,
\end{align*}
and
\[
  M_\alpha^{(2,1)}(t,\theta,\beta_1,\beta_2)
  = \prod_{i=0}^\infty
      \Bigl[
        1 - \alpha + \alpha\,J_{\alpha,i}^{(2,1)}(t,\theta,\beta_1,\beta_2)
      \Bigr],
\]
where
\begin{align*}
  J_{\alpha,i}^{(2,1)}(t,\theta,\beta_1,\beta_2)
  &= \frac{1}{i!}
     \int_{C_\alpha^{(2,1)}(t)}^\infty
       \ee^{-u}\,u^i\,
       \biggl[
         1 + \theta\,\frac{p_1}{p_2}\,
             \frac{(t/\pi\lambda_2)^{\beta_2}}
                  {(\alpha u/\pi\lambda_1)^{\beta_1}}
       \biggr]^{-\psi_1}
     \dd u,
\end{align*}
with
\[
  C_\alpha^{(2,1)}(t)
  = \frac{\pi\lambda_1}{\alpha}\,
     \biggl(
       \frac{b_1p_1}{b_2p_2}
     \biggr)^{1/\beta_1}\,
     \biggl(
       \frac{t}{\pi\lambda_2}
     \biggr)^{\beta_2/\beta_1}.
\]
$\tau(\theta,\beta)$ is also the same as in \eqref{eq:tau} with
$\Lpl_G(u^{-1})=\Lpl_{G,2}(u^{-1})=\bigl[u/(u+1)\bigr]^{\psi_2}$.
\end{theorem}

The proof is placed in the appendix and 
we here make a short remark on Theorem~\ref{thm:MIMO}.
The formula of the coverage probability in the theorem consists of two
parts \eqref{eq:MIMO1}--\eqref{eq:MIMO2} and
\eqref{eq:MIMO3}--\eqref{eq:MIMO4}.
The first part corresponds to that the typical user is served by a BS
of tier~$1$, so that the term in \eqref{eq:MIMO1} is given as the same
form as in \eqref{eq:coverage_GPP}.
The term in \eqref{eq:MIMO2} corresponds to the cumulative
interference from all the BSs of tier~$2$, which can be seen similar
to the second term in the exponential in \eqref{eq:coverage_PPP}.
The second part \eqref{eq:MIMO3}--\eqref{eq:MIMO4} corresponds to that
the typical user is served by a BS of tier~$2$, so that the term in
\eqref{eq:MIMO4} has the same form as the second term in the
exponential in \eqref{eq:coverage_PPP}.
The term in \eqref{eq:MIMO3} corresponds to the cumulative
interference from all the BSs of tier~$1$, so that $M_\alpha^{(2,1)}$
has a similar form to $M_\alpha$ in \eqref{eq:M} (the term
corresponding to $S_\alpha$ does not appear in this case).

\section{Conclusion}\label{sec:cncl}

In this tutorial, we have introduced the $\alpha$-Ginibre point
process as the model of BS deployments in cellular networks.
First, we have reviewed the definition and some useful properties of
this process, and then we have seen the two existing results on the
coverage analysis of cellular network models, where the BSs are
deployed according to the $\alpha$-Ginibre point processes.
The authors now hope that the readers will use the ($\alpha$-)Ginibre
point process and challenge themselves to various problems arising in
future cellular networks.

Finally, when we use the Ginibre and other determinantal point
processes as the models of BS deployments, we might face to a
computation problem.
Although the obtained formulas are indeed numerically computable, as
seen in \eqref{eq:coverage_GPP}--\eqref{eq:J} and
\eqref{eq:MIMO1}--\eqref{eq:MIMO4}, they include infinite sums and
infinite products, which may lead to the time-consuming computation.
One direction to avoid this problem could be some kinds of asymptotics
and/or approximation~(see, e.g.,
\cite{MiyoShir14b,NagaMiyoShir14,GantHaen16,WeiDengZhouHaen16,MiyoShir16}
for this direction).

\section*{Acknowledgments}
The first author's work was supported by the Japan Society for the
Promotion of Science (JSPS) Grant-in-Aid for Scientific Research (C)
16K00030.
The second author's work was supported by JSPS Grant-in-Aid for
Scientific Research (B) 26287019.


\appendix

\section{Proof of Theorem~\ref{thm:MIMO}}
We divide the coverage probability into two cases according to the
tier of the BS associated with the typical user;
\begin{align}\label{eq:MIMO_a1}
  \Prb\bigl(\SIR_o^{(\mathrm{MIMO})} > \theta_{\xi_{\eta(o)}}\bigr)
  &= \Prb\bigl(\SIR_o^{(\mathrm{MIMO})} > \theta_1,\: \xi_{\eta(o)}=1\bigr)
   + \Prb\bigl(\SIR_o^{(\mathrm{MIMO})} > \theta_2,\: \xi_{\eta(o)}=2\bigr),
\end{align}
and consider the two terms separately.

\paragraph{Case of $\xi_{\eta(o)}=1$:}
Let $\N_1$ and $\N_2$ denote the random partition of $\N$ such that
$\N_k = \{i\in\N: \xi_i=k\}$ for $k=1$, $2$.
Then, the interference~\eqref{eq:interference} for $i\in\N_1$ is
written as
\[
  I_o(i) = \sum_{j\in\N_1\setminus\{i\}}p_1\,G_j\,\ell_1(|X_j|)
         + \sum_{j\in\N_2}p_2\,G_j\,\ell_2(|X_j|).
\]
Applying this to the first term on the right-hand side of
\eqref{eq:MIMO_a1} yields
\begin{align}\label{eq:MIMO_a2}
  &\Prb\bigl(\SIR_o^{(\mathrm{MIMO})} > \theta_1,\: \xi_{\eta(o)}=1\bigr)
  \nonumber\\
  &= \Exp\Biggl[
       \exp\biggl\{
         -\theta_1\,
          \frac{I_o(\eta(o))}{p_1\,\ell_1(|X_{\eta(o)}|)}
       \biggr\}\,
       \ind{\{\xi_{\eta(o)}=1\}}
     \Biggr]
  \nonumber\\
  &= \Exp\Biggl[
       \prod_{j\in\N_1\setminus\{\eta(o)\}}
         \Lpl_{G,1}\Biggl(
           \theta_1\,
           \biggl(\frac{|X_{\eta(o)}|}{|X_j|}\biggr)^{2\beta_1}
         \Biggr)
       \prod_{j\in\N_2}
         \Lpl_{G,2}\Biggl(
           \theta_1\frac{p_2}{p_1}\frac{|X_{\eta(o)}|^{2\beta_1}}{|X_j|^{2\beta_2}}
         \Biggr)\,
       \ind{\{\xi_{\eta(o)}=1\}}
     \Biggr],
\end{align}
where $\ell_1(r)=r^{-2\beta_1}$ and $\ell_2(r)=r^{-2\beta_2}$ are also
used.
Note here that $\{\eta(o)=i\}$ with $i\in\N_1$ implies
$\{|X_j|\ge|X_i|\}$ for $j\in\N_1$ while for $j\in\N_2$,
\[
  \{b_1p_1|X_i|^{-2\beta_1} \ge b_2p_2|X_j|^{-2\beta_2}\}
  = \{|X_j|\ge R_{1,2}(|X_i|)\},
\]
with $R_{1,2}(r) = (b_2p_2/(b_1p_1))^{1/(2\beta_2)}r^{\beta_1/\beta_2}$.
Thus, using $\boldsymbol{Y}=\{Y_i\}_{i\in\N}$ in
Proposition~\ref{prp:Kostlan}, \eqref{eq:MIMO_a2} further reduces to
\begin{align}\label{eq:MIMO_a3}
  \Prb\bigl(\SIR_o^{(\mathrm{MIMO})} > \theta_1,\: \xi_{\eta(o)}=1\bigr)
  &= \alpha\sum_{i\in\N}
       \Exp\Biggl[
         \prod_{j\in\N\setminus\{i\}}
           \Biggl\{
             1 - \alpha
             + \alpha\,
               \Lpl_{G,1}\Biggl(
                 \theta_1\,\Biggl(\frac{Y_i}{Y_j}\Biggr)^{\beta_1}
               \Biggr)\,
               \ind{\{Y_j>Y_i\}}
           \Biggr\}
  \nonumber\\
  &\qquad\mbox{}\times
         \prod_{j\in\N}
           \Lpl_{G,2}\Biggl(
             \theta_1\frac{p_2}{p_1}\frac{{Y_i}^{\beta_1}}{|X_{2,j}|^{2\beta_2}}
           \Biggr)\,
           \ind{\{|X_{2,j}|\ge R_{1,2}({Y_i}^{1/2})\}}
       \Biggr],
\end{align}
where $\{X_{2,j}\}_{j\in\N}$ follows the homogeneous Poisson point
process with intensity~$\lambda_2$.
Conditioning on $Y_i$ and applying the generating
functional~\eqref{eq:Poisson_GF} to the second infinite product on the
right-hand side of \eqref{eq:MIMO_a3}, we have
\begin{align}\label{eq:MIMO_a4}
  &\Exp\Biggl[
     \prod_{j\in\N}
       \Lpl_{G,2}\Biggl(
         \theta_1\frac{p_2}{p_1}\frac{{Y_i}^{\beta_1}}{|X_{2,j}|^{2\beta_2}}
       \Biggr)\,
       \ind{\{|X_{2,j}|\ge R_{1,2}({Y_i}^{1/2})\}}
   \Biggm| Y_i \Biggr]
  \nonumber\\
  &= \exp\Biggl\{
       - \lambda_2\!\!
         \int_{\R^2}
           \Biggl[
             1 \!-\! \Lpl_{G,2}\Biggl(
                   \theta_1\frac{p_2}{p_1}\frac{{Y_i}^{\beta_1}}{|x|^{2\beta_2}}
                 \Biggr)
                 \ind{\{|x|\ge R_{1,2}({Y_i}^{1/2})\}}
           \Biggr]
         \dd x
     \Biggr\}
  \nonumber\\
  &= \exp\Biggl\{
       -\pi\lambda_2
        \biggl(\frac{b_2p_2}{b_1p_1}\biggr)^{1/\beta_2}
        {Y_i}^{\beta_1/\beta_2}\,
        \Bigl(1+\tau_{1,2}(\theta_1,\beta_2)\Bigr)
     \Biggr\},
\end{align}
where
\[
  \tau_{1,2}(\theta,\beta)
  = \frac{\theta^{1/\beta}}{\beta}
    \int_{1/\theta}^\infty
      \biggl[
        1 - \biggl(\frac{u}{u+b_1/b_2}\biggr)^{\psi_2}
      \biggr]\,
      u^{-1+1/\beta}\,
    \dd u.
\]
Hence, substituting \eqref{eq:MIMO_a4} to \eqref{eq:MIMO_a3} and
applying $Y_i\sim\Gam(i,\pi\lambda_1/\alpha)$, $i\in\N$, we obtain
\eqref{eq:MIMO1}--\eqref{eq:MIMO2} after some manipulations.

\paragraph{Case of $\xi_{\eta(o)}=2$:}
Similar to the above, the second term on the right-hand side of
\eqref{eq:MIMO_a1} is given as
\begin{align}\label{eq:MIMO_a5}
  &\Prb\bigl(\SIR_o^{(\mathrm{MIMO})} > \theta_2,\: \xi_{\eta(o)}=2\bigr)
  \nonumber\\
  &= \Exp\Biggl[
       \prod_{j\in\N_1}
         \Lpl_{G,1}\Biggl(
           \theta_2\,
           \frac{p_1}{p_2}
           \frac{|X_{\eta(o)}|^{2\beta_2}}{|X_j|^{2\beta_1}}
         \Biggr)
       \prod_{j\in\N_2\setminus\{\eta(o)\}}
         \Lpl_{G,2}\Biggl(
           \theta_2\,\biggl(\frac{|X_{\eta(o)}|}{|X_j|}\biggr)^{2\beta_2}
         \Biggr)\,
       \ind{\{\xi_{\eta(o)}=2\}}
     \Biggr].
\end{align}
Now, $\{\eta(o)=i\}$ with $i\in\N_2$ implies that
\[
  \{b_2p_2|X_i|^{-2\beta_2} \ge b_1p_1|X_j|^{-2\beta_1}\}
  = \{|X_j| \ge R_{2,1}(|X_i|)\},
\]
with $R_{2,1}(r) =
(b_1p_1/(b_2p_2))^{1/(2\beta_1)}r^{\beta_2/\beta_1}$ for $i\in\N_1$
while $\{|X_j|\ge|X_i|\}$ for $j\in\N_2$.
Therefore, using the distribution of $|X_{\eta(o)}|$ in
\eqref{eq:Poisson_nearest} and also $\boldsymbol{Y}=\{Y_i\}_{i\in\N}$ in
Proposition~\ref{prp:Kostlan}, \eqref{eq:MIMO_a5} reduces to
\begin{align}\label{eq:MIMO_a6}
  \Prb\bigl(\SIR_o^{(\mathrm{MIMO})} > \theta_2,\: \xi_{\eta(o)}=2\bigr)
  &= \int_0^\infty
       2\pi\lambda_2\,r\,\ee^{-\lambda_2\pi r^2}
       \Exp\Biggl[
         \prod_{j\in\N}
           \Biggl\{
             1 - \alpha
             + \alpha\,\Lpl_{G,1}\Biggl(
                 \theta_2\frac{p_1}{p_2}\frac{r^{2\beta_2}}{{Y_j}^{\beta_1}}
               \Biggr)\,
               \ind{\{Y_j\ge R_{2,1}(r)^2\}}
           \Biggr\}
       \Biggr]
  \nonumber\\
  &\qquad\mbox{}\times
       \Exp\Biggl[
         \prod_{j\in\N}
           \Lpl_{G,2}\Biggl(
             \theta_2\biggl(\frac{r}{|X_{2,j}|}\biggr)^{2\beta_2}
           \Biggr)
       \Biggm| |X_{2,j}|\ge r, j\in\N\Biggr]\,
     \dd r,
\end{align}
where by applying \eqref{eq:Poisson_GF}, the second expectation in the
integrand of \eqref{eq:MIMO_a6} is equal to
\begin{align*}
  \Exp\Biggl[
    \prod_{j\in\N}
      \Lpl_{G,2}\Biggl(
        \theta_2\biggl(\frac{r}{|X_{2,j}|}\biggr)^{2\beta_2}
      \Biggr)
  \Biggm| |X_{2,j}|\ge r, j\in\N\Biggr]
  &= \exp\Biggl\{
       -2\pi\lambda_2
        \int_r^\infty
          \Biggl[
            1 - \Lpl_{G,2}\Biggl(
                  \theta_2\biggl(\frac{r}{s}\biggr)^{2\beta_2}
                \Biggr)
          \Biggr]\,
        s\,\dd s
     \Biggr\}.
\end{align*}
Hence, substituting this to \eqref{eq:MIMO_a6} and applying
$Y_i\sim\Gam(i,\pi\lambda_1/\alpha)$, $i\in\N$, we obtain 
\eqref{eq:MIMO3}--\eqref{eq:MIMO4} after some manipulations.
\leavevmode\unskip\penalty9999\hbox{}\nobreak\hfill\quad\hbox{\qedsymbol}

\end{document}